\documentclass[a4paper, 11pt]{article}

\usepackage{amsmath,amsfonts,amssymb,amsthm,commath,graphicx, epstopdf, lmodern, txfonts}
\usepackage{graphicx, caption, subcaption}
\usepackage[latin1]{inputenc}
\usepackage[english]{babel}
\usepackage[colorlinks=true]{hyperref}
\usepackage{xspace}

\newcommand{\dotex}{\frac{d}{dt}}
\newcommand{\tr}[1]{\text{Tr}\left(#1\right)}

\newcommand{\Ktr}{\mathcal{K}^{1}(\mathcal{H})}
\newcommand{\Kf}{\mathcal{K}^{\text{f}}(\mathcal{H})}
\newcommand{\ket}[1]{\ensuremath{|#1\rangle}\xspace}
\newcommand{\bra}[1]{\ensuremath{\langle #1|}\xspace}

\newcommand{\q}[1]{\ensuremath{|#1\rangle}\xspace}
\newcommand{\qd}[1]{\ensuremath{\langle #1|}\xspace}

\newcommand{\ba}{\boldsymbol{a}}

\newcommand{\bA}{\boldsymbol{A}}

\newcommand{\bI}{\boldsymbol{I}}

\newcommand{\bL}{\boldsymbol{L}}

\newcommand{\bN}{{\boldsymbol{N}}}

\newcommand{\bP}{\boldsymbol{P}}
\newcommand{\bQ}{\boldsymbol{Q}}

\newcommand{\bX}{\boldsymbol{X}}

\newtheorem{theorem}{Theorem}
\newtheorem{lemma}{Lemma}
\newtheorem{prop}{Proposition}

\title{\LARGE \bf
Convergence and adiabatic elimination for a driven dissipative quantum harmonic oscillator\footnote{This work was partially supported by the Projet Blanc ANR-2011-BS01-017-01 EMAQS.}
}

\author{
R. Azouit \footnote{ Centre Automatique et Syst\`{e}mes, Mines-ParisTech, PSL Research University.
60, bd Saint-Michel 75006 Paris.} 
\and
A. Sarlette\footnote{INRIA Paris-Rocquencourt, Domaine de Voluceau, B.P. 105,
78153 Le Chesnay Cedex, France; and Ghent University / SYSTeMS, Technologiepark 914, 9052 Zwijnaarde, Belgium.}
\and
 P. Rouchon\footnotemark[2]
}

\begin{document}

\maketitle
\thispagestyle{empty}
\pagestyle{empty}

\begin{abstract}
We prove that a harmonic oscillator driven by Lindblad dynamics where the typical drive and loss channels are two-photon processes instead of single-photon ones, converges to a protected subspace spanned by two coherent states of opposite amplitude. We then characterize the slow dynamics induced by a perturbative single-photon loss on this protected subspace, by performing adiabatic elimination in the Lindbladian dynamics.
\end{abstract}

\section{Introduction}

The harmonic oscillator is a standard quantum system. It features coherent states, equivalent of classical harmonic oscillator amplitudes, whose coherent superpositions also called ``cat states'' feature genuinely quantum properties with no classical equivalent. In a recent paper \cite{MirrahimiCatComp2014}, our collaborators take advantage of this fact to propose an implementation of logical quantum bits as ``cat states'', with potential to realize universal quantum computation using standard technological elements as quantum gates. Their key contribution, besides the insight that the cat states are inherently insensitive to part of the typical perturbations, is the design of an ``engineered reservoir'' that stabilizes a subspace of such states in open loop.

The contribution of the present paper is to precisely establish, from the Lindblad master equation~\eqref{eq:system},  the stabilization properties of this scheme both, 
\begin{itemize}
\item  in ideal situations with $\epsilon=0$,  by proving in theorem~\ref{thm:convergence} global convergence of the infinite-dimensional nominal model towards the target ``protected subspace'', 
\item in presence of a small but dominant decoherence source with $\epsilon >0$,  by establishing the approximate   slow dynamics as a reduced Lindblad master equation~\eqref{eq:slow_sys}. 
\end{itemize}
The first goal builds on a typical Lyapunov-LaSalle strategy, with additional care needed due to the infinite dimension. For the second goal, we resort to a separation of the quantum dynamics into fast and slow components. We then apply an adiabatic elimination of the fast system to deduce a good approximation of the dynamics on a slow manifold, with the state remaining $\epsilon$-close to the protected subspace. Studying such perturbations is standard for quantum \emph{Hamiltonian} systems, where regular perturbation theory can be routinely applied \cite{sakurai2011modern}, but the Lindbladian case with singular perturbations has attracted much less attention. In \cite{mirrahimi-rouchonieee09} and similarly \cite{ReiteS2012PRA} singular perturbations up to second order are applied to a system with $N$ ground states and eliminating relaxing excited states. In \cite{atkins2003approximate,warszawski2000adiabatic} specific atom optics dynamics and an ancilla-mediated feedback are investigated with the standard approach of \cite{CarmichaelBook1993}. In \cite{Kessl2012PRA} the so-called Schrieffer-Wolff formalism is generalized to Lindbladian dynamics; its basic form requires inversion of the nominal dynamics operator, which is not too practical and which we circumvent here for the derivation of the reduced slow master equation~\eqref{eq:slow_sys}. 

The paper is organized as follows. Section II describes the mathematical model of the dynamics to be studied. Section III provides the global convergence proof for an idealized model, and Section IV analyzes precisely how this allows to counter the dominant external perturbation, which is single-photon loss. We pedagogically present the corresponding slow/fast perturbative argument (Section IV.B) as the translation to quantum notation of the standard dynamical systems approach, recalled in Section IV.A. Finally, simulations illustrate the validity of our analysis in Section V.

\section{Driven dissipative pairwise photon process}

The underlying space of the quantum harmonic oscillator is a Hilbert space $\mathcal{H}$ of infinite dimension spanned by the Fock states $\{ \ket{n} \}_{n \in\mathbb{N}}$. The annihilation operator $\ba$ is defined by $\ba\ket{n}=\sqrt{n}\ket{n-1}$ for any $n\geq 1$ and $\ba\ket{0}=0$. Its Hermitian conjugate $\ba^\dagger$, verifies $\ba^\dagger\ket{n}=\sqrt{n+1}\ket{n+1}$, for any $n\geq 0$; . We denote $\bN=\ba^\dagger \ba$ the photon number operator, satisfying $\bN\ket{n}=n \ket{n}$, for any $n\geq 0$. For any $\alpha \in \mathbb{C}$, a coherent state $\ket{\alpha}\in\mathcal{H}$ is characterized by
$\ba\ket{\alpha}=\alpha \ket{\alpha}$ and
\begin{equation*}
\ket{\alpha}=e^{-\frac{|\alpha|^2}{2}} \sum_{n=0}^\infty \frac{\alpha^n}{\sqrt{n!}}\ket{n}
.
\end{equation*}
The coherent states $\ket{\alpha}$ are viewed as the quantum model for a ``classical state'' of complex amplitude $\alpha$. Indeed, whereas the classical undamped harmonic oscillator $\tfrac{d}{dt} x = -\omega p$, $\tfrac{d}{dt} p = \omega x$ with $x,p \in \mathbb{R}$ has solutions $x(t) = \alpha e^{i \omega t}$, the quantum harmonic oscillator
$\tfrac{d}{dt}\q{\psi} = -i \omega (\bN+\bI/2) \q{\psi}$
features solutions $\q{\psi(t)} = \q{\alpha(t)} = \q{\alpha e^{i \omega t}}$ up to an irrelevant phase. As is customary, in the rest of the paper we describe the system in a frame rotating at the oscillator frequency $\omega$, for which the oscillator has stationary solutions $\q{\alpha}$.

We call Schr\"{o}dinger cat state or simply cat state the coherent superposition of two coherent states with opposite amplitudes,
\begin{equation} \label{eq:cat}
\ket{c_\alpha^{\pm}}=\frac{\ket{\alpha}\pm \ket{-\alpha}}{\mathcal{\gamma^\pm}}
\end{equation}
where $\gamma^\pm= \sqrt{2(1 \pm e^{-2|\alpha|^{2}})}$ is a normalization factor. If $|\alpha| \gg 1$ then we have $\gamma^+ \approx \gamma^- \approx \sqrt{2}$.

Throughout this paper, we denote by $\Ktr$ the set of  trace-class operators on $\mathcal{H}$, i.e. compact Hermitian operators on $\mathcal{H}$ whose eigenvalues $(\sigma_k)_{k\in\mathbb{N}}$ satisfy $\sum_{k\geq ^0} |\sigma_k| < +\infty$. This $\Ktr$ equipped with the trace-norm $\tr{|\cdot|}$ is a Banach space (see e.g.~\cite{TarasovBook2008}). The set of density operators (quantum states) corresponds to  elements of $\Ktr$ that are non-negative and  of trace $1$. They are usually  denoted by $\rho$. We denote also by $\Kf$, the subspace of $\Ktr$ of operators whose range is included in a vector space spanned by a finite number of Fock states:
$$
\Kf=\left\{ \sum_{\text{finite}} f_{n,n'} \ket{n}\bra{n'} : f_{n,n'}\in\mathbb{C}, f_{n,n'}=f^*_{n',n}\right\}
.
$$

We consider the  quantum  harmonic oscillator interacting with its environment as described in~\cite{MirrahimiCatComp2014}. An external coherent driving field of amplitude $u$ (assumed here real and strictly positive without loss of generality) is applied such that the oscillator can only exchange photons in pairs. Furthermore, the quantum system is built such that similarly the main dissipative process is a \emph{pairwise} photon loss with rate $\kappa >0$. Nevertheless, due to physical constraints, even if a single photon loss process is much less frequent than the previous one, it can't be totally omitted. For this reason, this term appears with a small coefficient $0< \epsilon  \ll \kappa $ in the following Lindblad master equation which governs the system dynamics:
\begin{equation*}
\dotex{\rho}=u[(\ba^\dagger)^2 - \ba^2,\rho] + \kappa \mathfrak{L}_{\ba^2}(\rho) + \epsilon \mathfrak{L}_{\ba}(\rho)
\end{equation*}
where $[\cdot,\cdot]$ stands for the commutator and where, for any  linear operator $\bA$  on $\mathcal{H}$, the super-operator  $\mathfrak{L}_{\bA}$  is given by
$$ \mathfrak{L}_{\bA}(\rho)= \bA \rho \bA^\dag- (\bA^\dag\bA\rho + \rho \bA^\dag\bA)/2 \, .$$
Noting $\alpha=2u/\kappa$ and $\bL=\ba^2-\alpha^2$ one can  reformulate the previous equation as :
\begin{equation} \label{eq:system}
\dotex{\rho}=\kappa\mathfrak{L}_{\bL}(\rho) + \epsilon \mathfrak{L}_{\ba}(\rho)
\end{equation}
It is shown in~\cite{MirrahimiCatComp2014} that, for $\epsilon=0$, the two-dimensional Hilbert space  $$
\mathcal{H}_\alpha=\text{span}\big\{\ket{\alpha},\ket{-\alpha}\big\}$$
   is a decoherence-free space:  for $\epsilon=0$, any density operator $\bar\rho$  with support included in $\mathcal{H}_\alpha $  is a steady state, i.e., $\mathfrak{L}_{\bL}(\bar\rho)=0$.

We will not investigate here  the well-posedness of~\eqref{eq:system} and the associated   strongly continuous semigroup of linear contraction on  $\Ktr$. Such issues can be investigated via theorem 3.1 of~\cite{Davie1977RoMP} ensuring the existence  of minimal solutions since  for any $\alpha,\kappa,\epsilon >0$, the operator
$-\kappa \bL^\dag \bL -\epsilon \ba^\dag \ba$ is the infinitesimal generator of a strongly continuous one parameter  contraction  semigroup  on $\mathcal{H}$. This means that, in the sequel, we will always assume that the  Cauchy problem~\eqref{eq:system} with an initial condition $\rho_0\in\Ktr$, positive semidefinite and of trace one, admits a solution in $\Ktr$ defined for any $t >0$. Moreover, this minimal solution will remain positive semidefinite.  We will investigate here the time asymptotic regime  for $\epsilon =0$ and then for $0 <\epsilon\ll\kappa $.

\section{Convergence of~\eqref{eq:system} for  $\epsilon=0$} \label{sec:conv}

\begin{lemma}\label{lem:Lyap}
 For any quantum state $\rho$ in $\Kf$ we have
 $
 \tr{\bL \mathfrak{L}_{\bL}(\rho) \bL^\dag} \leq -2 \tr{\bL \rho \bL^\dag}
 .
 $
\end{lemma}

\begin{proof}
From $
   \tr{\bL \mathfrak{L}_{\bL}(\rho) \bL^\dag}
   = \tr{\bL \rho \bL^\dag [\bL^\dag,\bL]}
$ and   $[\bL^\dag,\bL]=[(\ba^\dag)^2,\ba^2]= -4\bN-2\bI$, we have
  $$
 \tr{\bL \mathfrak{L}_{\bL}(\rho) \bL^\dag}
   = -2\tr{\bL \rho \bL^\dag(2\bN+\bI)}
.
$$
We conclude since   $\bN$ and $\bL\rho\bL^\dag$ are positive semidefinite.
\end{proof}

Modulo arguments proper to the infinite-dimensional setting, this basically means that $V(\rho) = \tr{\bL \rho \bL^\dag}$ is an exponential Lyapunov function for the system \eqref{eq:system} with $\epsilon=0$.

\begin{lemma}\label{lem:Nnu}
 For any $\nu\geq 1$ there exists $\mu\geq 0$ such that,  for any  quantum state $\rho$ in $\Kf$ we have
 $$\tr{\mathfrak{L}_{\bL}(\rho) \bN^\nu}\leq  - \nu  \left(\tr{\rho\bN^{\nu}}\right)^{\frac{\nu+1}{\nu}} + \mu
 $$
\end{lemma}

\begin{proof}
Denote by $\mathfrak{L}^*_{\bL}$ the adjoint of $\mathfrak{L}_{\bL}$:
$
\mathfrak{L}^*_{\bL}(\bA)= \bL^\dag \bA \bL- (\bL^\dag\bL\bA+ \bA \bL^\dag\bL)/2
$ for any Hermitian operator $\bA$ on $\mathcal{H}$.
Computations relying on the identity $\ba f(\bN) = f(\bN+\bI)\ba$ for any function $f$, yield
\begin{eqnarray*}
  \mathfrak{L}^*_{\bL}(f(\bN)) & = & - \bN(\bN-1) (f(\bN)-f(\bN-2\bI))
  \\
  && +\tfrac{\alpha^2}{2} \ba^2(f(\bN)-f(\bN-2\bI))  
  +\tfrac{\alpha^2}{2} (f(\bN)-f(\bN-2\bI))(\ba^\dag)^2 \; .
\end{eqnarray*}
Since $
\tr{\mathfrak{L}_{\bL}(\rho) f(\bN)} = \tr{\rho \mathfrak{L}^*_{\bL}(f(\bN)) }$, we have
\begin{eqnarray*}
  \tr{\mathfrak{L}_{\bL}(\rho) f(\bN)}
  & = &  
  - \tr{\rho \bN(\bN-1) (f(\bN)-f(\bN-2\bI))}
  \\
  && +\tfrac{\alpha^2}{2} \tr{\rho\ba^2(f(\bN)-f(\bN-2\bI))}
  \\
  && +\tfrac{\alpha^2}{2}\tr{\rho (f(\bN)-f(\bN-2\bI))(\ba^\dag)^2}
  .
\end{eqnarray*}
Take $f(x)=x^\nu$ and define $g(x)=f(x)-f(x-2)$ for $x\geq 2$; $g(x)=f(x)$ for $2>x \geq 0$; $g(x)=f(0)
$ for $x<0$.
From $\ba^2 g(\bN)= \ba\sqrt{g(\bN-\bI)} \ba \sqrt{g(\bN)} $ and Cauchy Schwartz inequality
\begin{eqnarray*}
 |\tr{\rho\ba^2 g(\bN)}|
  & = & \left|\tr{\big(\sqrt{\rho} \ba \sqrt{g(\bN-\bI)} \big)\big( \ba\sqrt{g(\bN)} \sqrt{\rho}\big)}\right|
  \\
  &\leq & \sqrt{\tr{\rho g(\bN-2\bI)(\bN+\bI)}\tr{\rho g(\bN) \bN}} \, .
\end{eqnarray*}
Since  $g(x-2)(x+1) \leq g(x)(x+3)$ for $x\geq 0$,  we have
$$
\big| \tr{\rho\ba^2 g(\bN)}\big| \leq \tr{\rho g(\bN) (\bN+3\bI)}
 .
$$
Thus
$$
\tfrac{1}{2}\tr{\rho (\ba^2g(\bN)+g(\bN)(\ba^\dag)^2)}\leq \tr{\rho g(\bN) (\bN+3\bI)}
$$
and
$$
  \tr{\mathfrak{L}_{\bL}(\rho) \bN^\nu}  
  \leq  \tr{\rho g(\bN)\left(-\bN^2+ (\alpha^2+1)\bN  + 3\alpha^2 \right)}
  .
$$
Since $(x^\nu -(x-2)^\nu)\left(-x^2+ (\alpha^2+1)x + 3\alpha^2 \right)$ is equivalent to $-2\nu x^{\nu+1}$ for large $x$, there exists $\mu >0$ such that for all $x \geq 0$,
$$
((x+2)^\nu -x^\nu)\left(-x^2+ (\alpha^2+1)x + 3\alpha^2 \right)
\leq - \nu x^{\nu+1} + \mu
.
$$
Finally we get
$$
 \tr{\mathfrak{L}_{\bL}(\rho) \bN^\nu}
  \leq  - \nu  \tr{\rho\bN^{\nu+1}} + \mu
$$
Since $ x^{\frac{\nu}{\nu+1}}$ is concave,
$
\left(\tr{\rho\bN^{\nu+1}} \right)^{\frac{\nu}{\nu+1}} \geq \tr{\rho \bN^\nu}
.
$
 \end{proof}\vspace{3mm}

\begin{theorem} \label{thm:convergence}
Consider a trajectory  $[0,+\infty[\ni t\mapsto \rho(t)\in\Ktr$ of  the  master equation~\eqref{eq:system} with $\epsilon=0$, $\kappa >0$ and $\alpha >0$. The following statements hold true:
\begin{enumerate}
  \item  \label{thm:convergence1 } Take $\nu\geq 1$. Then there exists $\gamma >0$ such that,  for any  initial quantum state $\rho(0)=\rho_0$ satisfying $\tr{\rho_0\bN^\nu}< +\infty$,  we have  for all $t>0$, $\tr{\rho(t) \bN^{\nu}} \leq \max\big(\gamma, \tr{\rho_0\bN^\nu}\big)$.
\item \label{thm:convergence2} Assume that $\rho_0\in\Kf$. Then, there exists a quantum state  $\bar\rho$  with  support in $\mathcal{H}_\alpha$ such that, for any $\nu \geq 0$,  $\lim_{t\mapsto +\infty}\tr{\left| \bN^{\frac{\nu}{2}} (\rho(t)-\bar\rho )\bN^{\frac{\nu}{2}} \right|}=0$.
\end{enumerate}
\end{theorem}
 The limit $\bar\rho$ depends on  $\rho_0$. It can be derived from $\rho_0$ with the four Hermitian, bounded and independent  operators that are in the kernel of the adjoint super-operator $\mathfrak{L}^*_{\bL}$ and given in~\cite{MirrahimiCatComp2014}.

\begin{proof}
The first statement is a direct consequence of Lemma~\ref{lem:Nnu} and of the fact that quantum states element of   $\Kf$ are a dense subset of the quantum states $\sigma$ of $\Ktr$ with $\tr{\sigma\bN^{\nu}}$ finite.  $\tr{\rho(t) \bN^\nu}$   remains bounded since
$$
\dotex \tr{\rho \bN^\nu} = \kappa\tr{\mathfrak{L}_{\bL}(\rho) \bN^\nu} \leq - v \kappa \tr{\rho \bN^{\nu}}^{\frac{\nu+1}{\nu}} + \kappa \mu
.
$$
Thus for $\tr{\rho\bN^\nu} \geq \lambda= \left(\tfrac{\mu}{\nu}\right)^{\frac{\nu}{\nu+1}}$, $\dotex \tr{\rho\bN^\nu} \leq 0$.

The second statement exploits the first one. We can assume $\nu \geq 2$. Take $\nu' >\nu$.
Denote by $\mathcal{K}^{1}_{\nu'}(\mathcal{H})$  the supspace of trace-class operators $\sigma$ such that $\tr{\left| \bN^{\frac{\nu'}{2}}\sigma \bN^{\frac{\nu'}{2}} \right|}$ is finite. The space  $\mathcal{K}^{1}_{\nu'}(\mathcal{H})$ with the norm
$\| \sigma \|_{\nu'} = \tr{|\sigma|} + \tr{\left| \bN^{\frac{\nu'}{2}}\sigma \bN^{\frac{\nu'}{2}} \right|}$ is a Banach space.
From the first statement,  we know that $\rho_0$ being an element of $\mathcal{K}^{1}_{\nu'}(\mathcal{H})$, $\rho(t)$ remains always in $\mathcal{K}^{1}_{\nu'}(\mathcal{H})$. Since  $\nu'>\nu$,  the injection of  $\mathcal{K}^{1}_{\nu'}(\mathcal{H})$ into $\mathcal{K}^{1}_{\nu}(\mathcal{H})$ is compact: $\{\rho(t)~|~t\geq 0\}$ is precompact in $\mathcal{K}^{1}_{\nu}(\mathcal{H})$.
Denote by $\bar\rho\in \mathcal{K}^{1}_{\nu}(\mathcal{H})$ an adherent point of $\rho(t)$ for $t$ tending towards infinity. Since $\nu \geq 2$, $\bar\rho$ and $\rho$ belong to the domain of $\mathfrak{L}_{\bL}$.  Lemma~\ref{lem:Lyap}  implies   $\tr{\bL \bar\rho \bL^\dag}=0$, i.e., the support of $\bar\rho$ is contained in the kernel of $\bL$, which coincides with $\mathcal{H}_\alpha$.
Moreover, the semigroup associated to the Lindblad master equation is a contraction for the trace distance: for two trajectories  $\rho_1(t)$ and $\rho_2(t)$,  $t\mapsto \tr{|\rho_1(t)-\rho_2(t)|}$ is a non-increasing function. Thus $t\mapsto \tr{|\rho(t)-\bar\rho|} $ is non-increasing  since $\bar\rho$ is a steady state. Consequently the adherent point $\bar\rho$ is unique: $\rho(t)$ converges towards $\bar\rho$ in $\mathcal{K}^{1}_{\nu}(\mathcal{H})$.
\end{proof}

\section{Reduced slow dynamics of~\eqref{eq:system}}

We have proved in the previous section that the system converges toward the decoherence free  subspace $\mathcal{H}_\alpha$ when we neglect the photon loss channel ($\epsilon =0$). When $0 < \epsilon \ll 1$, the center manifold theorem allows us to separate the system into fast and slow dynamics. The fast dynamics makes the system globally  converge to a subspace close to $\mathcal{H}_\alpha$. The slow dynamics approximate the behavior of ``protected states'' $\q{c^+_\alpha},\q{c^-_\alpha}$. The present section is aimed at  characterizing these dynamics to the first order in $\epsilon$. 

The fast/slow dynamics reduction from nonlinear systems theory, also known as singular perturbation theory, is useful despite the \emph{linearity} of Lindbladian dynamics, because the very high dimension and often high degeneracy of open quantum systems makes matrix diagonalization impractical to apply.

For the sake of clarity, we first particularize the fast/slow dynamics reduction theory to linear systems of finite dimension with the standard notations. We propose then a reduction procedure via an adapted duality viewpoint with the characterization based on~\eqref{eq:der1} and~\eqref{eq:der2} here below. We apply then this characterization to the  quantum system~\eqref{eq:system} and show that it facilitates the computations up to first order with respect to e.g.~\cite{Kessl2012PRA}.


\subsection{Reducing a linear system to its slow dynamics}\label{ssec:lin}

We here review the theory of geometric singular perturbation, which was mainly developed by Fenichel in \cite{Fenichel79} and surveyed by Jones in \cite{Jones1995}, in a linear context.
Consider the nominal linear system $\dotex{x} = A x$ with $x=(x_1,x_2) \in \mathbb{R}^{m}\times\mathbb{R}^{n}$ and converging globally to the $m$-dimensional subspace $\mathcal{S} = \{ x \in \mathbb{R}^{m+n} : x_2 \in \mathbb{R}^n = 0 \}$. 
To this nominal dynamics we add an arbitrary perturbation matrix $B$ of order $\epsilon \ll 1$. In matrix notation, the dynamics can be written in block form which yields:
\begin{eqnarray}\label{eq:OrLinSyst}
\tfrac{d}{dt}x_1 & = & A_1 x_2 + \epsilon (B_1 x_2 + B_0 x_1) \\
\nonumber \tfrac{d}{dt}x_2 & = & A_2 x_2 + \epsilon (B_2 x_2 + B_3 x_1) \, .
\end{eqnarray}
The assumption that $\mathcal{S}$ is globally exponentially stable for $\epsilon=0$ corresponds to $A_2$ having all eigenvalues with strictly negative real parts, thus it is invertible. Hence we can define the regular change of variables
\begin{equation}
\begin{aligned}
&\tilde{x}_1 = x_1 - A_1 A_2^{-1} x_2 \\
&\tilde{x}_2 = x_2
\end{aligned}
\end{equation}
which yields dynamics in Tikhonov normal form
\begin{eqnarray}\label{eq:TikSlow}
\tfrac{d}{dt} \tilde{x}_1 & = & \epsilon \left( (B_0-A_1A_2^{-1} B_3) \tilde{x}_1 \right. \\
\nonumber && \left. + (B_0+B_1 - A_1 A_2^{-1}(B_2+B_3)) \tilde{x}_2 \vphantom{A_1^{-1}} \right) \\
\nonumber & = & \epsilon f(\tilde{x}_1,\tilde{x}_2) \, .\\
\label{eq:TikFast}
\tfrac{d}{dt} \tilde{x}_2 & = & A_2 \tilde{x}_2 + \epsilon (B_3 \tilde{x}_1 + (B_2+B_3A_1A_2^{-1}) \tilde{x}_2)\, \\
\nonumber & = & g(\tilde{x}_1, \tilde{x}_2, \epsilon) \, .
\end{eqnarray}
The Tikhonov conditions for reducing the system by singular perturbations is that the first (slow) subsystem has eigenvalues going down as $\epsilon$, while the second (fast) subsystem has eigenvalues bounded away from zero for $\epsilon = 0$. These conditions are satisfied above. The Tikhonov theorem then allows the following reduction.
\begin{prop}
The trajectories of the full system \eqref{eq:TikSlow},\eqref{eq:TikFast} (with initial conditions satisfying $g= 0$) 
remain $\epsilon$-close over at least a time of order $1/\epsilon$, to the trajectories of the system restricted to the ``slow submanifold'' $g(\tilde{x}_1, \tilde{x}_2,0) = 0$ and where dynamics are given by replacing $\tilde{x}_2$ in $f(\tilde{x}_1,\tilde{x}_2)$ by the solution of $g(\tilde{x}_1, \tilde{x}_2,0) = 0$.
\end{prop}
In our linear case, the slow manifold comes down to $\tilde{x}_2 = 0$ and the slow dynamics trivially reduce to the first term in \eqref{eq:TikSlow}. Transforming back to the original coordinates the slow manifold corresponds just to $x_2=0$ and the dynamics are
\begin{equation}\label{eq:LinTheDyns}
\tfrac{d}{dt}x_1 = \epsilon (B_0-A_1A_2^{-1} B_3) x_1 \, .
\end{equation}
The second term reflects the influence of the fast $x_2$ dynamics on the slow variable $x_1$: by blindly setting $x_2=0$ in the original system \eqref{eq:OrLinSyst} and neglecting its second line, we would miss this term and get an incorrect approximation.\\

Computing the corrective term $A_1A_2^{-1} B_3$ by explicit inversion of $A_2$ can be a tedious task when the fast subsystem has a large dimension (in our quantum case, $x_2$ would rigorously be of infinite dimension). However if first integrals of the system with $\epsilon=0$ are known, these can facilitate the computations via the following dual viewpoint.

Consider a linear functional $p^T = (p_1^T , \; p_2^T) \in \mathbb{R}^{*\,m+n}$ which is conserved by $\dot{x} = A x$, i.e.~satisfying $\dot{p} = A^T p = 0$ or equivalently $p_1^T A_1 x_2 + p_2^T A_2 x_2 = 0$
for all $x_2 \in \mathbb{R}^n$. (The notation $\cdot^T$ denotes matrix transpose.) Again using invertibility of $A_2$, we see that $p^T$ actually satisfies
\begin{equation}\label{eq:functional}
p_1^T A_1 A_2^{-1} y + p_2^T y = 0 \; \forall y \in \mathbb{R}^n \, .
\end{equation}
Knowing $m$ linearly independent functionals $(p^T(k))_{k=1,\ldots,m}$ satisfying \eqref{eq:functional} is sufficient to fully characterize the corrective term in \eqref{eq:LinTheDyns}: we have
\begin{equation}\label{eq:der1}
\tfrac{d}{dt}x_1 = \epsilon (B_0+Q) x_1
\end{equation}
with $Q$ defined by the set of linear equations:
\begin{equation}\label{eq:der2}
p_1^T(k) Q = p_2^T(k) B_3 \, ,\; k=1,2,...,m \, .
\end{equation}


\subsection{Quantum system~\eqref{eq:system} with $0<\epsilon\ll\kappa$}\label{ssec:quantred}

We know apply the same procedure to our quantum system. The nominal $\dot{x} = A x$ corresponds to $\dot{\rho} = \mathfrak{L}_{\bL}(\rho)$.
We have shown in Section \ref{sec:conv} that:
\begin{itemize}
\item $\mathcal{S}$ corresponds to a four-dimensional real subspace of Hermitian operators spanned by $\q{c^+_\alpha}\qd{c^+_\alpha}$, $\q{c^-_\alpha}\qd{c^-_\alpha}$, $\q{c^+_\alpha}\qd{c^-_\alpha}+\q{c^-_\alpha}\qd{c^+_\alpha}$, and $i(\q{c^+_\alpha}\qd{c^-_\alpha}-\q{c^-_\alpha}\qd{c^+_\alpha})$. These correspond to the $m=4$ coordinates of $x_1$.
\item The subspace $\mathcal{S}$ is globally asymptotically stable under $\dot{\rho} = \mathfrak{L}_{\bL}(\rho)$. This corresponds to the invertibility condition on $A_2$ independently of $\epsilon$.
\end{itemize}
We therefore introduce the projector
\begin{equation}
\bP_c = \ket{c^+_\alpha}\bra{c^+_\alpha}+\ket{c^-_\alpha}\bra{c^-_\alpha}
\end{equation}
such that
\begin{equation}
\rho_s = \bP_c \rho \bP_c
\end{equation}
corresponds to the ``slow'' $x_1$ space of the previous section. The projection onto ``$x_2$ space'' is given by
$$\rho_f = \rho - \bP_c \rho \bP_c \, .$$
Our goal is to compute the evolution of $\rho_s$, which is the equivalent of \eqref{eq:LinTheDyns}. For this we take advantage of the dual formulation \eqref{eq:der1},\eqref{eq:der2}. The following procedure can in principle be applied to any perturbative dynamics, we here focus on $\mathfrak{L}_{\ba}$ as a physically relevant case.

The perturbative dynamics on the slow manifold features a first component, corresponding to $B_0$, obtained simply by projection onto the slow manifold. The identities $\ba \q{c^\pm} = \alpha \frac{\gamma^{\mp}}{\gamma^\pm} \q{c^\mp}$ quickly yield its explicit expression:
\begin{eqnarray}
\bP_c \mathfrak{L}_{\ba} (\rho_s) \bP_c &\!\! = \!\!& \alpha^2 \mathfrak{L}_{\bX}(\rho_s) \\
\text{where }\bX &\!\! = \!\!& \frac{\gamma_+}{\gamma_-} \ket{c^+}\bra{c^-}+ \frac{\gamma_-}{\gamma_+} \ket{c^-}\bra{c^+} \, .\phantom{KK}
\end{eqnarray}

To compute the corrective term by duality, using the equivalent of \eqref{eq:der1},\eqref{eq:der2}, we need to identify $m=4$ conserved functionals of the system, which are the fixed points of the dual nominal dynamics $\dotex{\xi} = \mathfrak{L}^*_{\bL}(\xi)$. Fortunately, those invariants are known for the particular operator~$\bL$:
\begin{itemize}
\item One easily checks that $ \xi^a =\bI$ the identity operator  is in the kernel of any $ \mathfrak{L}^*_{\bL}$.
\item The parity operator $\xi^b= (-1)^{\ba^\dagger \ba}$ is  in the kernel of $ \mathfrak{L}^*_{\bL}$ because photons are exchanged by pairs.
\item The appendix of \cite{MirrahimiCatComp2014} gives two more operators $\xi^c$ and $\xi^d$ in terms of Bessel functions; one checks that they are linearly independent for finite $\alpha$.
\end{itemize}
We will also use the following key property of the conserved quantities, which is specific to the structure of quantum Lindblad dynamics.
\begin{lemma}\label{lem1}
Any Hermitian operator $\xi$ in $\ker(\mathfrak{L}^*_L)$ commutes with $\bP_c$ the orthogonal projector onto $\mathcal{H}_\alpha$. 
\end{lemma}
\begin{proof} From $2 \bL^\dag\xi \bL = \bL^\dag \bL \xi + \xi \bL^\dag \bL$ and $\bL \bP_c=0=\bP_c\bL^\dag$ we have 
$\bL^\dag \bL \xi \bP_c=0=\bP_c\xi \bL^\dag \bL$. Thus the Hermitian operator $\bA=\bP_c \xi + \xi\bP_c$ satisfies 
$\bL\bA \bL^\dag = (\bL^\dag \bL \bA + \bA \bL^\dag \bL)/2$, i.e. belongs to  the kernel of $\mathfrak{L}_{\bL}$. The support of $\bA$ is thus included in $\mathcal{H}_\alpha$ and thus $[\bP_c,\bA]=0$. This implies that $\bP_c \xi = \bP_c \xi \bP_c = \xi \bP_c$. 
\end{proof}

Considering projections on respective subspaces, the equivalent of the $B_3$ term of the linear perturbation is given by
$\mathfrak{L}_{\ba}(\rho_s) - \bP_c \mathfrak{L}_{\ba}(\rho_s) \bP_c$. The equivalent of equation~\eqref{eq:der2}  characterizes the Hermitian  operator $\bQ$ with support on $\mathcal{H}_\alpha$ as follows: 
\begin{eqnarray*}
\tr{\bP_c \xi^\nu \bP_c \bQ}
&=& \tr{\left(\xi^\nu - \bP_c \xi^\nu \bP_c\right) \left( \mathfrak{L}_a(\rho_s) - \bP_c \mathfrak{L}_a(\rho_s) \bP_c \right)}
\\&=& \tr{\xi^\nu \left( \mathfrak{L}_a(\rho_s) - \bP_c \mathfrak{L}_a(\rho_s) \bP_c \right)}
\\&=& \tfrac{1}{2} \tr{\xi^\nu \left((\bP_c-I) a^\dagger a \rho_s + \rho_s a^\dagger a (\bP_c-I)\right)}
\\&=&  \tfrac{1}{2} \tr{\xi^\nu \left(a^\dagger a \rho_s (\bP_c-I) + (\bP_c-I) \rho_s a^\dagger a \right)}
\;\; =0 \, 
\end{eqnarray*}
for any $\nu=a,b,c,d$.
From the first to the second line, $\bP_c \xi^\nu \bP_c$ is readily dropped since $\bP_c$ is a projector. For the next one we have to write out $\mathfrak{L}_{\ba}$ and see that $\bP_c \ba \rho_s \ba^\dagger \bP_c = \ba \rho_s \ba^\dagger$ for the particular perturbation operator $\ba$. The last line follows by using Lemma \ref{lem1} and $\bP_c \rho_s = \rho_s$.
As a conclusion, we get that the corrective term is $\bQ=0$ for our particular case. We can summarize these computations as follows.\vspace{2mm}  

\emph{The trajectories of system \eqref{eq:system} with initial conditions  having supports in $\mathcal{H}_\alpha$  remain $\epsilon$-close over at least a time of order $1/\epsilon$, to the trajectories of the ``slow variable'' $\rho_s$ which is a linear combination of $\q{c^+_\alpha}\qd{c^+_\alpha}$, $\q{c^-_\alpha}\qd{c^-_\alpha}$, $\q{c^+_\alpha}\qd{c^-_\alpha}+\q{c^-_\alpha}\qd{c^+_\alpha}$, and $i(\q{c^+_\alpha}\qd{c^-_\alpha}-\q{c^-_\alpha}\qd{c^+_\alpha})$. The slow state $\rho_s$ follows the Lindblad dynamics:
\begin{eqnarray}
\tfrac{d}{dt} \rho_s &\!\! = \!\!& \epsilon \alpha^2 \mathfrak{L}_{\bX}(\rho_s) \label{eq:slow_sys} \\
\text{where } \bX &\!\! = \!\!& \frac{\gamma_+}{\gamma_-} \ket{c^+}\bra{c^-}+ \frac{\gamma_-}{\gamma_+} \ket{c^-}\bra{c^+} \, .\phantom{KK} \label{eq:slow_sys_def}
\end{eqnarray}}\vspace{2mm}

{In applications considering $\q{c^+_\alpha}$ and $\q{c^-_\alpha}$ as canonical states $\q{0},\q{1}$ of a logical qubit \cite{MirrahimiCatComp2014}, the operator $X$ corresponds to a bit-flip in the limit $\frac{\gamma_+}{\gamma_-} \rightarrow 1$ of large coherent amplitude $\alpha$ and to a decoherence to the vacuum $\q{0}$ in the limit of Fock states $\q{c^+_\alpha}=\q{n=0}$, $\q{c^-_\alpha}=\q{n=1}$ when $\alpha=0$. For all other cases, the qubit dynamics \eqref{eq:slow_sys},\eqref{eq:slow_sys_def} corresponds on the canonical Bloch sphere to:
\begin{eqnarray*}
\tfrac{d}{dt} x & = & -\alpha^2 \; \frac{(\gamma_+^2-\gamma_-^2)^2}{2\gamma_+^2\gamma_-^2} \; x    \\
\tfrac{d}{dt} y & = & -\alpha^2 \; \frac{(\gamma_+^2+\gamma_-^2)^2}{2\gamma_+^2\gamma_-^2} \; y   \\
\tfrac{d}{dt} z & = & -\alpha^2 \; \frac{\gamma_+^4 + \gamma_-^4}{\gamma_+^2\gamma_-^2} \;  \left(z-\tfrac{\gamma_+^4-\gamma_-^4}{\gamma_+^4+\gamma_-^4}\right) \; .
\end{eqnarray*}
This converges to $x=0$ (slowly for $\frac{\gamma_+}{\gamma_-} \simeq 1$), $y=0$, and $z=\tfrac{\gamma_+^4-\gamma_-^4}{\gamma_+^4+\gamma_-^4}$ (which is $\simeq 0$ for $\frac{\gamma_+}{\gamma_-} \simeq 1$).}

\section{Numerical simulations}

To illustrate the interest of this model reduction based on adiabatic elimination of the rapidly converging variables, we compare  via numerical simulations trajectories of the complete system and of the reduced one. To do so, we use a numerical scheme which preserves the positiveness for the Lindblad equation and similar to one used in~\cite{LeR2013PRA}.  We choose the following values of the parameters : a time-step of $10^{-3}$, the decoherence strength $\kappa=1$,  $u=1/2$ and $\epsilon=0.01$. With  $\alpha=1$,  the population of photons for $n  > n_{max}=40$ in the coherent state $\ket{\alpha}$ is almost zero since less that $\frac{1}{n_{max}!}$. Consequently, we truncate the infinite-dimensional Hilbert space to a numerical system space spanned by $\{ \ket{1},\ket{2}, \dots \ket{40}\}$ in the Fock basis. We denote by $\rho$ the density matrix of the resulting system. It is worth stressing that the complete system is represented by an $n_{max}\times n_{max}$ matrix while the reduced system is represented by a $2 \times 2$ matrix (on the basis $\ket{c^+_\alpha},\ket{c^-_\alpha})$. Thus the computation is much faster on the second one.

We first take as initial condition the vacuum state, $\rho_0=\ket{0}\bra{0}$. The state of the reduced system, $\rho_s$ is then initialized at $\ket{c^+_\alpha}\bra{c^+_\alpha}$ because both $\rho_0$ and $\ket{c^+_\alpha}\bra{c^+_\alpha}$ are $+1$ eigenstates of the parity operator $\xi^b= (-1)^{\ba^\dagger \ba}$, which is a conserved quantity (see Section \ref{ssec:quantred} and~\cite{MirrahimiCatComp2014}). To compare the trajectories of~\eqref{eq:system} initialized at $\ket{0}\bra{0}$ and of~\eqref{eq:slow_sys} initialized at $\ket{c^+_\alpha}\bra{c^+_\alpha}$, we show in figure~\ref{fig:comp} the expectation values $\tr{\rho \sigma_z}$ and $\tr{\rho_s \sigma_z}$ of the operator $\sigma_z=\ket{c^+_\alpha}\bra{c^+_\alpha}-\ket{c^-_\alpha}\bra{c^-_\alpha}$, commonly denoted $\langle \sigma_z \rangle$.
After a transitional regime of typical duration $1/\kappa$, one can see a strong similarity between $\tr{\rho \sigma_z}$ and $\tr{\rho_s \sigma_z}$ up to a constant offset. The value of this offset is of order $\epsilon$.
Furthermore, we plot the fidelity $F(\rho,\rho_s)=tr \left( \sqrt{\sqrt{\rho_s}\rho \sqrt{\rho_s}} \right)$ between $\rho_s$ and  $\rho$.
For better readability, figure \ref{fig:fidelity} shows the logarithm of 1 minus the fidelity, i.e.~of its deviation from the ideal value 1. This deviation quickly converges to an order $10^{-4}$, corresponding to $\epsilon^2$ as expected. It then further decreases, incidentally, as both systems converge towards the unique equilibrium of the slow dynamics.

To emphasize the influence of $\gamma^+$ and $\gamma^-$ in \eqref{eq:slow_sys_def}, we add a simulation with the same parameters but with the following and same initial condition for the complete and reduced system:
\begin{equation*}
\tilde{\rho}_0= \frac{1}{2}\left(\ket{c^+_\alpha}+\ket{c^-_\alpha}\right)\left(\bra{c^+_\alpha}+\bra{c^-_\alpha}\right)
\end{equation*}
Figure \ref{fig:comp2_x} shows that the expectation value of $\sigma_x= \ket{c^+_\alpha}\bra{c^-_\alpha}+\ket{c^-_\alpha}\bra{c^+_\alpha}$ slowly decreases over time, as expected from bit-flip dynamics. The slope of this decrease is approximated to $\sim 4\%$ accuracy by the reduced dynamics. Moreover, figure \ref{fig:comp2} establishes that $\langle \sigma_z \rangle$ does not remain zero. This is due to the fact that, with $\gamma^+ > \gamma^-$, equation \eqref{eq:slow_sys} ``promotes'' the population of $\ket{c^+_\alpha}\bra{c^+_\alpha}$ over the population of $\ket{c^-_\alpha}\bra{c^-_\alpha}$, unlike a pure bit-flip.

The simulations thus confirm the validity of our approximation of the complete model by the reduced one.

\begin{figure}
\begin{center}
\includegraphics[scale=0.55]{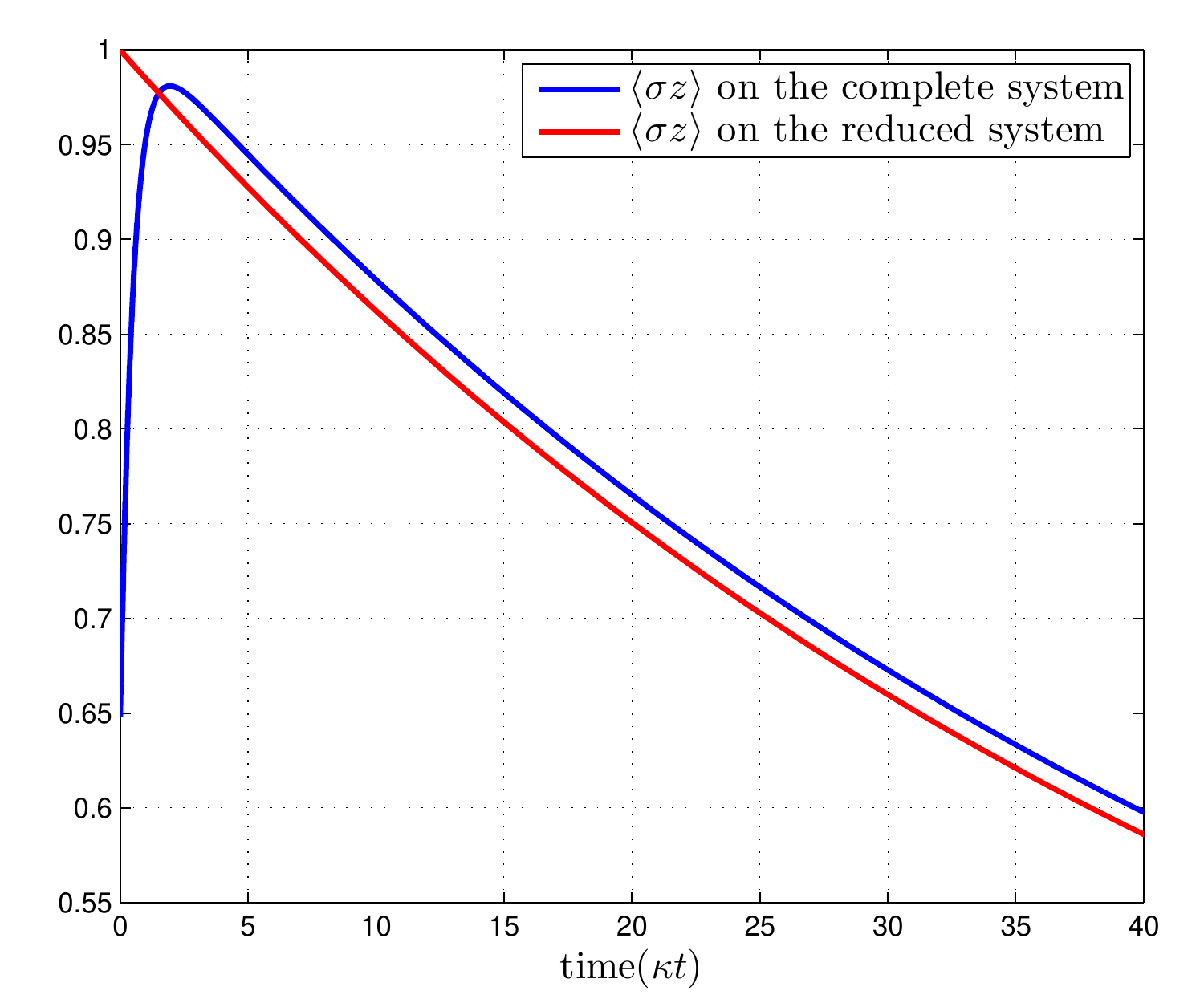}
\end{center}
\caption{Comparison of $\tr{\sigma_z \rho}$ and $\tr{\sigma_z \rho_s}$  for $\rho$ solution of  the complete system~\eqref{eq:system} with $\epsilon=\tfrac{1}{100}$, $\alpha=1$ and vacuum  initial condition (truncation up to $40$ photons),  and for $\rho_s$ solution of  the reduced system~\eqref{eq:slow_sys} with initial condition $\ket{c^+_\alpha}\bra{c^+_\alpha}$.}\label{fig:comp}
\end{figure}

\begin{figure}
\begin{center}
\includegraphics[scale=0.55]{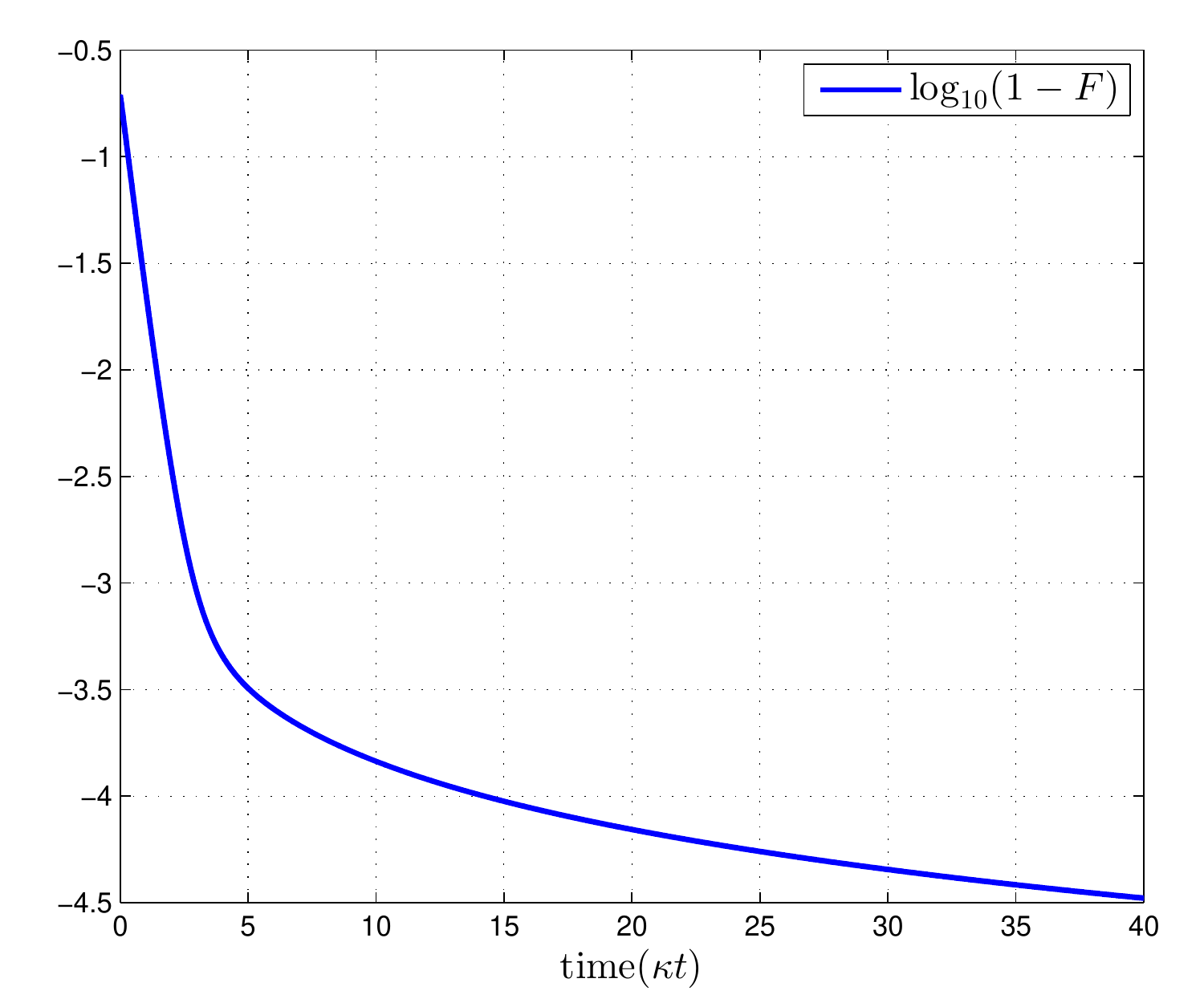}
\end{center}
\caption{$\log_{10}(1-F)$ where $F$ is the fidelity between $\rho_s$ and $\rho$  for the simulations of figure~\ref{fig:comp}.}
\label{fig:fidelity}
\end{figure}

\begin{figure}
\begin{center}
\includegraphics[scale=0.55]{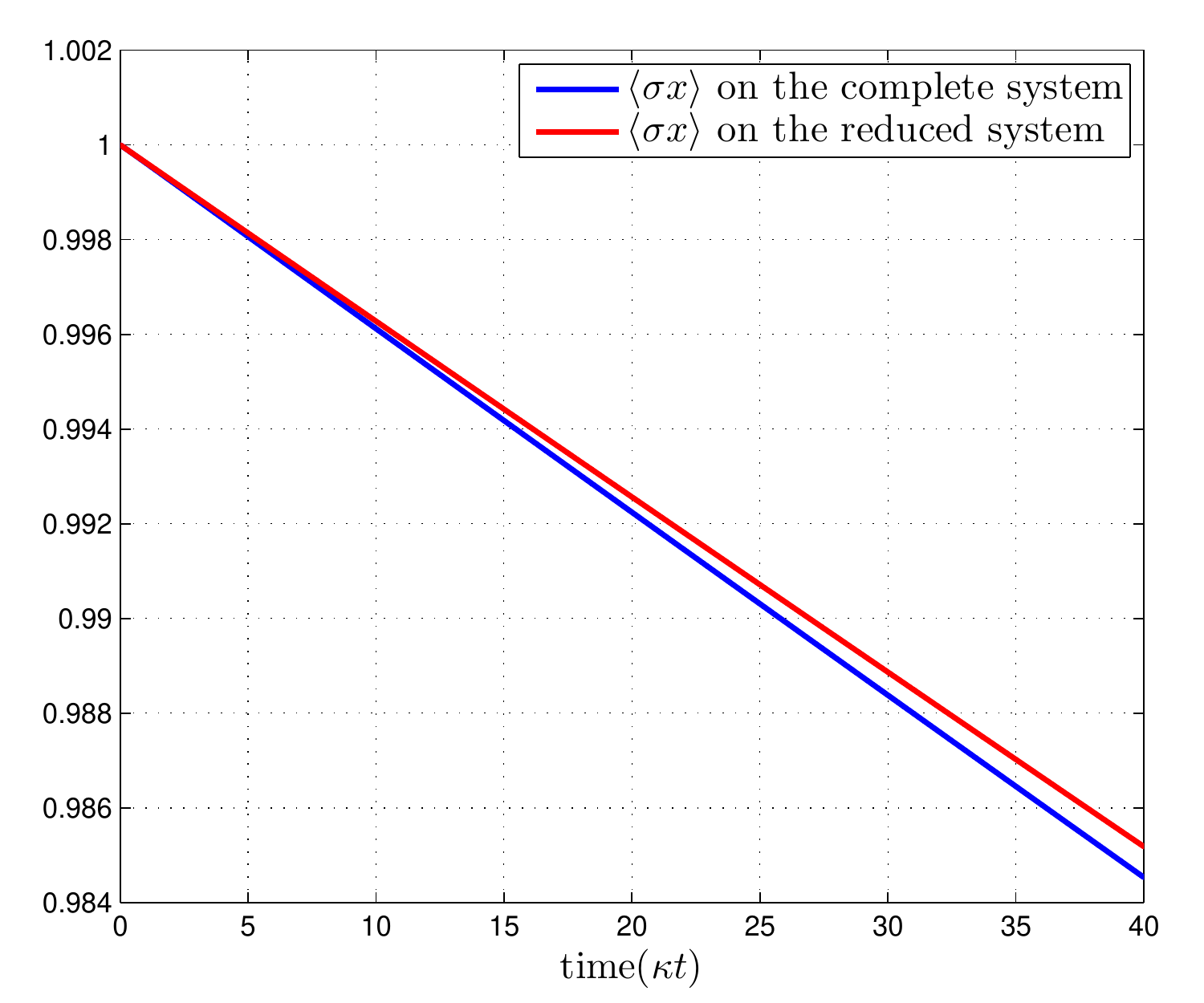}
\end{center}
\caption{Comparison of $\tr{\sigma_x \rho}$ and $\tr{\sigma_x \rho_s}$  for $\rho$ solution of  the complete system~\eqref{eq:system} with $\epsilon=\tfrac{1}{100}$, $\alpha=1$  (truncation up to $40$ photons), and for $\rho_s$ solution of  the reduced system~\eqref{eq:slow_sys}, with the same  initial condition $\rho(0)=\rho_s(0)=\frac{1}{2}\left(\ket{c^+_\alpha}+\ket{c^-_\alpha}\right)\left(\bra{c^+_\alpha}+\bra{c^-_\alpha}\right)$.}\label{fig:comp2_x}
\end{figure}

\begin{figure}
\begin{center}
\includegraphics[scale=0.55]{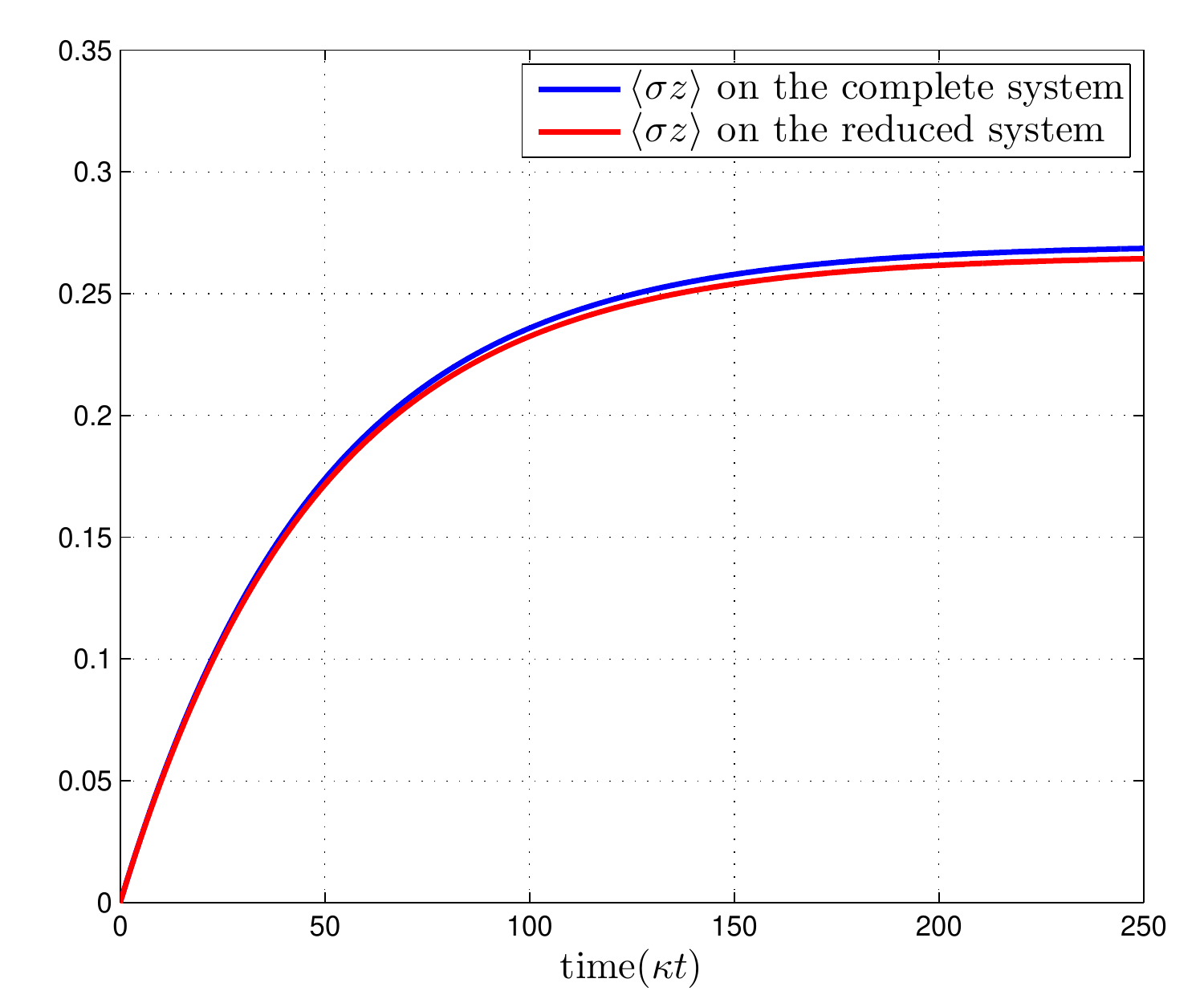}
\end{center}
\caption{Comparison of $\tr{\sigma_z \rho}$ and $\tr{\sigma_z \rho_s}$ where $\rho$ and $\rho_s$ correspond to simulations of figure~\ref{fig:comp2_x}.}\label{fig:comp2}
\end{figure}

\section{CONCLUSIONS} 

We have rigorously proved convergence of a harmonic oscillator Lindblad dynamics with two-photon exchanges, to a protected subspace. We have also established the approximate slow dynamics on this protected subspace when a typical perturbation is added, and illustrated its validity in simulations. The methods used for this particular example are applicable to general Lindbladian dynamics.

The reduction by singular perturbations and adiabatic elimination is of course applicable in general to evaluate the remaining slow dynamics in quantum systems with (engineered) protected subspaces. Extension to $k$-photon processes $\ba^k$ with $k>2$ can be addressed in the same way. The fact that the slow variable still follows a Lindbladian master equation may not be surprising but remains to be proved in the general case. The fact that the dynamics reduces to the orthogonal projection of the Lindbladian onto the protected subspace (i.e.~$B_0$ without any correction due to $B_3$, in the terms of Section \ref{ssec:lin}) for the case examined here would be in agreement with the physicists' ``quantum Zeno'' viewpoint. However, under which formulation this viewpoint should be applied in the general case also remains to be rigorously characterized.

\addtolength{\textheight}{-12cm}   

\section*{ACKNOWLEDGMENT}
The authors thank Mazyar Mirrahimi for many useful discussions.

\bibliographystyle{plain}

\end{document}